\documentclass{sig-alternate}

\usepackage{amsthm,amssymb,amsmath}
\usepackage{graphics}
\usepackage{tikz}
\usepackage[plainpages=false,pdfpagelabels,colorlinks=true,citecolor=blue,hypertexnames=false]{hyperref}
\usepackage{color}
\usepackage{enumitem}

\CopyrightYear{2017}
\setcopyright{acmlicensed}
\conferenceinfo{ISSAC '17,}{July 25-28, 2017, Kaiserslautern, Germany}
\isbn{978-1-4503-5064-8/17/07}\acmPrice{\$15.00}
\doi{http://dx.doi.org/10.1145/3087604.3087616}

\newtheorem{theorem}{Theorem}
\newtheorem{proposition}[theorem]{Proposition}
\newtheorem{corollary}[theorem]{Corollary}
\newtheorem{lemma}[theorem]{Lemma}

\newtheorem{conjecture}[theorem]{Conjecture}

\newtheorem{example}[theorem]{Example}

\def\qed{\quad\rule{1ex}{1ex}}

\theoremstyle{definition}
\newtheorem{remark}[theorem]{Remark}

\let\set\mathbb
\DeclareMathOperator{\ord}{ord}

\DeclareMathOperator{\lc}{lc}
\def\O{\mathrm{O}}
\DeclareMathOperator{\Ker}{ker}
\DeclareMathOperator{\rem}{rem}
\DeclareMathOperator{\nonrem}{nrem}
\DeclareMathOperator{\LCLM}{lclm}
\DeclareMathOperator{\Op}{Op}
\DeclareMathOperator{\NOp}{NOp}
\DeclareMathOperator{\In}{in}
\DeclareMathOperator{\Cl}{Cl}
\DeclareMathOperator{\Res}{Res}
\DeclareMathOperator{\corank}{corank}

\begin{document}




\title{Bounds for Substituting Algebraic Functions into~D-finite~Functions}

\numberofauthors{2}
\author{%
 \alignauthor
 \leavevmode
 \mathstrut Manuel Kauers\titlenote{Supported by the Austrian Science Fund (FWF): Y464, F5004.}\\[\smallskipamount]
  \affaddr{\leavevmode\mathstrut Institute for Algebra / Johannes Kepler University}\\
  \affaddr{\leavevmode\mathstrut 4040 Linz, Austria}\\
  \affaddr{\leavevmode\mathstrut manuel.kauers@jku.at}
 \and
 \mathstrut Gleb Pogudin\titlenote{Supported by the Austrian Science Fund (FWF): Y464.}\\[\smallskipamount]
  \affaddr{\leavevmode\mathstrut Institute for Algebra / Johannes Kepler University}\\
  \affaddr{\leavevmode\mathstrut 4040 Linz, Austria}\\
  \affaddr{\leavevmode\mathstrut gleb.pogudin@jku.at}
}

\maketitle

\begin{abstract}
  It is well known that the composition of a D-finite function with an algebraic function is again D-finite.
  We give the first estimates for the orders and the degrees of annihilating operators for the compositions.
  We find that the analysis of removable singularities leads to an order-degree curve which is much more
  accurate than the order-degree curve obtained from the usual linear algebra reasoning. 
\end{abstract}

\section{Introduction}

A function $f$ is called D-finite if it satisfies an ordinary linear differential equation 
with polynomial coefficients, 
\[
 p_0(x)f(x)+p_1(x)f'(x)+\cdots+p_r(x)f^{(r)}(x)=0.
\]
A function $g$ is called algebraic if it satisfies a polynomial equation with polynomial
coefficients,
\[
 p_0(x)+p_1(x)g(x) + \cdots + p_r(x)g(x)^r=0.
\]
It is well known~\cite{stanley80} that when $f$ is D-finite and $g$ is algebraic,
the composition $f\circ g$ is again D-finite. For the special case $f=\mathrm{id}$ this reduces to Abel's theorem, which says
that every algebraic function is D-finite. This particular case was investigated closely
in~\cite{bostan07}, where a collection of bounds was given for the orders and degrees of the
differential equations satisfied by a given algebraic function. It was also pointed out in~\cite{bostan07}
that differential equations of higher order may have significantly lower degrees,
an observation that gave rise to a more efficient algorithm for transforming an algebraic equation
into a differential equation. Their observation has also motivated the study of order-degree
curves: for a fixed D-finite function~$f$, these curves describe the boundary of the region of all pairs $(r,d)\in\set N^2$
such that $f$ satisfies a differential equation of order~$r$ and degree~$d$.

\begin{example}\label{ex:1} We have fixed some randomly chosen operator
  $L\in C[x][\partial]$ of order $r_L=3$ and degree $d_L=4$
  and a random \break
  \null

  \kern-1.7\baselineskip
  \noindent\parbox[b]{.6\hsize}{%
   polynomial $P\in C[x][y]$ of $y$-degree $r_P=3$ and $x$-degree $d_P=4$.
  For some prescribed orders~$r$, we computed the smallest degrees~$d$ such that there is an operator $M$
  of order $r$ and degree $d$ that annihilates $f\circ g$ for all solutions $f$ of $L$
  and all solutions $g$ of~$P$.
  The points $(r,d)$ are shown in the figure on the right.}\kern-.25em
  \raisebox{-.6em}
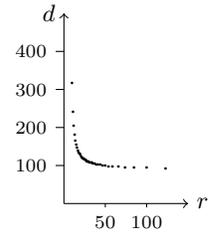
{%
    \begin{tikzpicture}[yscale=.46,scale=1.1]
      \draw[->] (0,0)--(0,5) node[left] {$d$};
      \draw[->] (0,0)--(1.5,0) node[right] {$r$};
      \foreach\x in {1,...,4} \draw (0,\x)--(-.1,\x) node[left] {\scriptsize$\x 00$};
      \foreach\x in {1,...,1} \draw (\x,0)--(\x,-.1) node[below] {\scriptsize$\x 00$};
      \draw (.5,0)--(.5,-.1) node[below] {\scriptsize$50$};
      \clip (-1,-1) rectangle (1.5,5);
      \draw
(.10,3.16) node {$\cdot$}
(.11,2.40) node {$\cdot$}
(.12,2.02) node {$\cdot$}
(.13,1.79) node {$\cdot$}
(.14,1.64) node {$\cdot$}
(.15,1.53) node {$\cdot$}
(.16,1.45) node {$\cdot$}
(.17,1.38) node {$\cdot$}
(.18,1.33) node {$\cdot$}
(.19,1.29) node {$\cdot$}
(.20,1.26) node {$\cdot$}
(.21,1.23) node {$\cdot$}
(.22,1.20) node {$\cdot$}
(.23,1.18) node {$\cdot$}
(.24,1.16) node {$\cdot$}
(.25,1.14) node {$\cdot$}
(.26,1.13) node {$\cdot$}
(.27,1.12) node {$\cdot$}
(.28,1.10) node {$\cdot$}
(.29,1.09) node {$\cdot$}
(.30,1.08) node {$\cdot$}
(.31,1.07) node {$\cdot$}
(.33,1.06) node {$\cdot$}
(.34,1.05) node {$\cdot$}
(.35,1.04) node {$\cdot$}
(.37,1.03) node {$\cdot$}
(.39,1.02) node {$\cdot$}
(.41,1.01) node {$\cdot$}
(.44,1.00) node {$\cdot$}
(.47,.99) node {$\cdot$}
(.50,.98) node {$\cdot$}
(.54,.97) node{$\cdot$}
(.59,.96) node{$\cdot$}
      (.66,.95) node{$\cdot$}
      (.74,.94) node{$\cdot$}      
      (.85,.93) node{$\cdot$}   
      (1.00,.92) node{$\cdot$}   
      (1.23,.91) node{$\cdot$}
      (1.61,.90) node{$\cdot$}
      ;
    \end{tikzpicture}}
\end{example}

Experiments suggested that order-degree curves are often just simple hyperbolas. \emph{A priori} knowledge of these
hyperbolas can be used to design efficient algorithms. For the case of creative 
telescoping of hyperexponential functions and hypergeometric terms, as well as for 
simple D-finite closure properties (addition, multiplication, Ore-action), bounds for order-degree curves
have been derived~\cite{chen12b,chen12c,kauers14f}. However, it turned out that these bounds are often not
tight. 

A new approach to order-degree curves has been suggested in~\cite{jaroschek13a},
where a connection was established between order-degree curves and apparent
singularities. Using the main result of this paper, very accurate order-degree
curves for a function~$f$ can be written down in terms of the number and the
cost of the apparent singularities of the minimal order annihilating operator
for~$f$. However, when the task is to compute an annihilating operator
from some other representation, e.g., a definite integral, then the information
about the apparent singularities of the minimal order operator is only \emph{a
  posteriori} knowledge. Therefore, in order to design efficient algorithms using the result of \cite{jaroschek13a},
we need to predict the singularity structure of the output operator in terms of
the input data. This is the program for the present paper.

First (Section~\ref{sec:linear_algebra}), we derive an order-degree bound for
D-finite substitution using the classical approach of considering a suitable
ansatz over the constant field, comparing coefficients, and balancing variables
and equations in the resulting linear system. This leads to an order-degree
curve which is not tight. Then (Section~\ref{sec:min_operator}) we estimate the
order and degree of the minimal order annihilating operator for the composition
by generalizing the corresponding result of~\cite{bostan07} from $f=\mathrm{id}$
to arbitrary D-finite~$f$. The derivation of the bound is a bit more tricky in
this more general situation, but once it is available, most of the subsequent
algorithmic considerations of~\cite{bostan07} generalize
straightforwardly. Finally (Section~\ref{sec:3}) we turn to the analysis of the
singularity structure, which indeed leads to much more accurate results. The
derivation is also much more straightforward, except for the required
justification of the desingularization cost. In practice, it is almost always
equal to one, and although this is the value to be expected for generic input,
it is surprisingly cumbersome to give a rigorous proof for this expectation.

Throughout the paper, we use the following conventions: 
\begin{itemize}
\item $C$ is a field of characteristic zero, $C[x]$ is the usual commutative ring of 
   univariate polynomials over~$C$. We write $C[x][y]$ or $C[x,y]$ for the commutative 
   ring of bivariate polynomials and $C[x][\partial]$ for the non-commu\-ta\-tive ring of linear
   differential operators with polynomial coefficients. In this latter ring, the multiplication
   is governed by the commutation rule $\partial x=x\partial +1$.

\kern-\medskipamount
\item $L\in C[x][\partial ]$ is an operator of order $r_L:=\deg_{\partial }(L)$ 
   with polynomial coefficients of degree at most~$d_L:=\deg_x(L)$.

\kern-\medskipamount
\item $P\in C[x,y]$ is a polynomial of degrees $r_P:=\deg_y(P)$ and $d_P:=\deg_x(P)$.
   It is assumed that $P$ is square-free as an element of $C(x)[y]$ and that 
   it has no divisors in $\bar C[y]$, where $\bar C$ is the algebraic closure of~$C$.

\kern-\medskipamount
\item $M\in C[x][\partial ]$ is an operator such that for every solution $f$ of~$L$
   and every solution $g$ of~$P$, the composition $f\circ g$ is a solution of~$M$.
   The expression $f \circ g$ can be understood either as a composition of analytic functions
   in the case $C = \mathbb{C}$, or in the following sense.
   We define $M$ such that for every $\alpha \in C$, for every solution $g \in C[[x - \alpha]]$ of $P$
   and every solution $f \in C[[x - g(\alpha)]]$ of $L$, $M$ annihilates $f\circ g$, which is a well-defined
   element of $C[[x - \alpha]]$. In the case $C = \mathbb{C}$ these two definitions coincide.
\end{itemize}


\section{Order-Degree-Curve\hfill\break by Linear Algebra}\label{sec:linear_algebra}

Let $g$ be a solution of~$P$, i.e., suppose that $P(x,g(x))=0$, and let $f$ be a
solution of~$L$, i.e., suppose that $L(f)=0$. Expressions involving $g$ and $f$ can
be manipulated according to the following three well-known observation:
\begin{enumerate}
\item\label{ob:1} (Reduction by~$P$) For each polynomial $Q\in C[x,y]$ with $\deg_y(Q)\geq r_P$
  there exists a polynomial $\tilde Q\in C[x,y]$ with $\deg_y(\tilde Q)\leq\deg_y(Q)-1$ and $\deg_x(\tilde Q)\leq\deg_x(Q)+d_P$
  such that
  \[
    Q(x,g)= \frac1{\lc_y(P)}\tilde Q(x,g).
  \]
  The polynomial $\tilde{Q}$ is the result of the first step of computing the pseudoremainder of $Q$ by $P$ w.r.t.~$y$.
\item\label{ob:2} (Reduction by~$L$) There exist polynomials $v,q_{j,k}\in C[x]$ of degree at most $d_Ld_P$ such that
  \[
    f^{(r_L)}\circ g=\frac1{v}\sum_{j=0}^{r_P-1}\sum_{k=0}^{r_L-1} q_{j,k}g^j \cdot (f^{(k)}\circ g).
  \]
  To see this, write $L=\sum_{k=0}^{r_L}l_k\partial ^k$ for some polynomials $l_k\in C[x]$ of degree at most~$d_L$. Then we have
  \[
    f^{(r_L)}\circ g = \frac{-1}{l_{r_L}\circ g}\sum_{k=0}^{r_L-1}(l_k\circ g) \cdot (f^{(k)}\circ g).
  \]
  By the assumptions on~$P$, the denominator $l_{r_L}\circ g$ cannot be zero. In other words, $\gcd(P(x,y),l_{r_L}(y))=1$ in $C(x)[y]$.
  For each $k=0,\dots,r_L-1$, consider an ansatz $AP+Bl_{r_L}=l_k$ for polynomials $A,B\in C(x)[y]$ of
  degrees at most $d_L-1$ and $r_P-1$, respectively, and compare coefficients with respect to~$y$.
  This gives $k$ inhomogeneous linear systems over $C(x)$ with $r_P+d_L$ variables and equations, which
  only differ in the inhomogeneous part but have the same matrix $M=\operatorname{Syl}_y(P,l_{r_L})$ for every~$k$.   
  The claim follows using Cramer's rule, taking into account that the coefficient matrix of the system
  has $d_L$ many columns with polynomials of degree~$d_P$ and $r_P$ many columns with polynomials of degree $\deg_x l_k(y)=0$
  (which is also the degree of the inhomogeneous part). Note that $v=\det(M)$ does not depend on~$k$.
\item\label{ob:3} (Multiplication by~$g'$) For each polynomial $Q\in C[x,y]$ with $\deg_y(Q)\leq r_P-1$
  there exist polynomials $q_j\in C[x]$ of degree at most $\deg_x(Q)+2r_Pd_P$ such that
  \[
    g' Q(x,g) = \frac1{w\lc_y(P)}\sum_{j=0}^{r_P-1}q_jg^j,
  \]
  where $w\in C[x]$ is the discriminant of~$P$.
  To see this, first apply Observation~\ref{ob:1} (Reduction by $P$) to rewrite $-QP_x$ as $T=\frac1{\lc_y(P)}\sum_{j=0}^{2r_P-2}t_j y^j$
  for some $t_j\in C[x]$ of degree $\deg_x(Q)+d_P$.
  Then consider an ansatz $AP+BP_y=\lc_y(P)T$ with unknown polynomials $A,B\in C(x)[y]$ of degrees
  at most $r_P-2$ and $r_P-1$, respectively, and compare coefficients with respect to~$y$.
  This gives an inhomogeneous linear system over $C(x)$ with $2r_P-1$ variables and equations.
  The claim then follows using Cramer's rule. 
\end{enumerate}

\begin{lemma}\label{lemma:x}
  Let $u = vw\lc_y(P)^{r_P}$, where $v$ and $w$ are as in the Observations \ref{ob:2} and~\ref{ob:3} above.
  Let $f$ be a solution of $L$ and $g$ be a solution of~$P$.
  Then for every $\ell\in\set N$ there are polynomials $e_{i,j}\in C[x]$ of degree at most $\ell\deg(u)$ such that
  \[
  \partial ^\ell (f\circ g) = \frac1{u^\ell}\sum_{i=0}^{r_P-1}\sum_{j=0}^{r_L-1} e_{i,j} g^i \cdot (f^{(j)}\circ g).
  \]
\end{lemma}
\begin{proof}
This is evidently true for $\ell=0$. Suppose it is true for some~$\ell$. Then
\begin{alignat*}1
  &\partial ^{\ell+1} (f\circ g) =
  \sum_{i=0}^{r_P-1}\sum_{j=0}^{r_L-1} \biggl(\frac{e_{i,j}}{u^\ell} g^i \cdot (f^{(j)}\circ g)\biggr)'\\
  &=
  \sum_{i=0}^{r_P-1}\sum_{j=0}^{r_L-1} \biggl(
  \frac{e_{i,j}'u - \ell e_{i,j}u'}{u^{\ell+1}} g^i \cdot (f^{(j)}\circ g) \\
  &\qquad{}+
  \frac{e_{i,j}}{u^\ell}\Bigl(i\,g^{i-1} \cdot (f^{(j)}\circ g) + g^i \cdot (f^{(j+1)}\circ g)\Bigr) g'
  \biggr).
\end{alignat*}
The first term in the summand expression already matches the claimed bound. To complete the proof, we
show that
\begin{equation}\label{eq:5}
\bigl(i\,g^{i-1} \cdot (f^{(j)}\circ g) + g^i \cdot (f^{(j+1)}\circ g)\bigr) g'=\frac1{u}\sum_{k=0}^{r_P-1}q_kg^k
\end{equation}
for some polynomials~$q_k$ of degree at most~$\deg(u)$. Indeed, the only critical term is $f^{(r_L)}\circ g$. According to
Observation~\ref{ob:2}, $f^{(r_L)}\circ g$ can be rewritten as $\frac1{v}\sum_{j=0}^{r_P-1}\sum_{k=0}^{r_L-1}q_{j,k}g^j \cdot (f^{(k)}\circ g)$
for some $q_{j,k}\in C[x]$ of degree at most~$d_Ld_P$. This turns the left hand side of~\eqref{eq:5} into an expression
of the form $\frac1{v}\sum_{j=0}^{2r_P-2}\tilde q_{j,k}g^j \cdot (f^{(k)}\circ g)$ for some polynomials
$\tilde q_{j,k}\in C[x]$ of degree at most~$d_Ld_P$. An $(r_P-1)$-fold application of Observation~\ref{ob:1} 
brings this expression to the form $\frac1{v\lc_y (P)^{r_P-1}}\sum_{j=0}^{r_P-1}\bar q_{j,k}g^j \cdot (f^{(k)}\circ g)$
for some polynomials $\bar q_{j,k}\in C[x]$ of degree at most $d_Ld_P+(r_P-1)d_P$. Now Observation~\ref{ob:3}
completes the induction argument.
\end{proof}

\begin{theorem}\label{thm:linalgcurve}
Let $r,d\in\set N$ be such that
\[
r \geq r_Lr_P
\quad\text{and}\quad
d \geq \frac{r(3r_P+d_L-1)d_P r_L r_P}{r+1-r_L r_P}.
\]
Then there exists an operator $M\in C[x][\partial ]$ of order $\leq r$ and degree $\leq d$
such that for every solution~$g$ of~$P$ and every solution~$f$ of~$L$ the composition
$f\circ g$ is a solution of~$M$.
In particular, there is an operator $M$ of order $r=r_Lr_P$ and degree
$(3r_P+d_L-1)d_P r_L^2 r_P^2 = \O\bigl((r_P+d_L)d_P r_L^2 r_P^2\bigr)$.
\end{theorem}

\begin{proof}
Let $g$ be a solution of $P$ and $f$ be a solution of $L$. 
Then we have $P(x,g(x))=0$ and $L(f) = 0$, and we seek an operator
$M=\sum_{i=0}^d\sum_{j=0}^r c_{i,j} x^i \partial ^j\in C[x][\partial ]$ such that
$M(f\circ g)=0$. Let $r\geq r_Lr_P$ and consider an ansatz
\[
  M=\sum_{i=0}^d\sum_{j=0}^r c_{i,j}x^i\partial ^j
\]
with undetermined coefficients $c_{i,j}\in C$.

Let $u$ be as in Lemma~\ref{lemma:x}. Then applying $M$ to $f\circ g$ and multiplying by $u^r$ gives an expression
of the form
\[
\sum_{i=0}^{d+r\deg(u)}\sum_{j=0}^{r_P-1}\sum_{k=0}^{r_L-1} q_{i,j,k}x^i g^j \cdot (f^{(k)}\circ g),
\]
where
the $q_{i,j,k}$ are $C$-linear combinations of the undetermined coefficients~$c_{i,j}$. Equating all the $q_{i,j,k}$
to zero leads to a linear system over $C$ with at most $(1+d+r\deg(u))r_Lr_P$ equations and exactly $(r+1)(d+1)$
variables. This system has a nontrivial solution as soon as
\begin{alignat*}3
       &&  (r+1)(d+1) &> (1+d+r\deg(u))r_L r_P \\
  \iff &&  (r+1-r_L r_P)(d+1) &> r\,r_Lr_P\deg(u) \\
  \iff &&  d & > -1 + \frac{r\,r_Lr_P\deg(u)}{r+1-r_Lr_P}.
\end{alignat*}
The claim follows because $\deg(u)\leq d_Pd_L+(2r_P-1)d_P+r_Pd_P=(3r_P+d_L-1)d_P$. 
\end{proof}


\section{A Degree Bound for\hfill\break the Minimal Operator}\label{sec:min_operator}

According to Theorem~\ref{thm:linalgcurve}, there is
operator $M$ of order $r=r_Lr_P$ and degree $d=\O((r_P+d_L)d_Pr_L^2r_P^2)$.
Usually there is no operator of order less than $r_Lr_P$, but if such an operator
accidentally exists, Theorem~\ref{thm:linalgcurve} makes no statement about its degree.
The result of the present section (Theorem~\ref{thm:min_operator} below) is a degree bound
for the minimal order operator, which also applies when its order is less than~$r_Lr_P$,
and which is better than the bound of Theorem~\ref{thm:linalgcurve} if the minimal order
operator has order~$r_Lr_P$.

The following Lemma is a variant of Lemma~\ref{lemma:x} in which $g$ is allowed to appear
in the denominator, and with exponents larger than $r_P-1$. This allows us to keep the
$x$-degrees smaller. 

\begin{lemma}\label{lemma:derivative_of_composition}
  Let $f$ be a solution of $L$ and $g$ be a solution of~$P$.
  For every $\ell \in \set N$, there exist polynomials $E_{\ell, j} \in C[x, y]$ for $0 \leq j < r_L$ 
  such that $\deg_x E_{\ell, j} \leq \ell(2d_P - 1)$ and $\deg_y E_{\ell, j} \leq \ell(2r_P+ d_L - 1)$ for all $0 \leq j < r_L$, and
  $$
    \partial^\ell \left( f \circ g\right) = \frac{1}{U(x, g)^\ell} \sum\limits_{j = 0}^{r_L - 1} E_{\ell,j}(x, g) (f^{(j)} \circ g),
  $$
  where $U(x, y) = P_y^2(x, y) l_{r_L}(y)$.
\end{lemma}

\begin{proof}
  This is true for $\ell = 0$. Suppose it is true for some $\ell$. Then
  \begin{multline*}
    \partial^{\ell + 1}(f \circ g) = \left( \frac{1}{U(x, g)^\ell} \sum\limits_{j = 0}^{r_L - 1} E_{\ell, j}(x, g) (f^{(j)} \circ g) \right)^{\prime} \\
    = \sum\limits_{j = 0}^{r_L - 1} \Bigl( \frac{\ell(U_x + g' U_y)}{U^{\ell + 1}} E_{i, j} \cdot (f^{(j)} \circ g) \\
    + \frac{1}{U^{\ell}}((E_{\ell, j})_x + g' \cdot (E_{\ell, j})_y) (f^{(j)} \circ g) + \frac{1}{U^{\ell}} E_{\ell, j} g' \cdot (f^{(j + 1)} \circ g) \Bigr)
  \end{multline*}

  We consider the summands separately. 
  In $\frac{\ell(U_x + g' U_y)}{U^{\ell + 1}}$, $U_x$ is already a polynomial in $x$ and $g$ of bidegree at most $(2d_p - 1, 2r_P + d_L - 1)$.
  Since $g' = \frac{-P_x(x, g)}{P_y(x, g)}$ and $U_y$ is divisible by $P_y$, $g' U_y$ is also a polynomial with the same bound for the bidegree.
 
  Futhermore, we can write
  \[
    (E_{\ell, j})_x + g' \cdot (E_{\ell, j})_y = \frac{1}{U} (U(E_{\ell, j})_x - P_xP_yl_{r_L}(g)(E_{\ell, j})_y),
  \]
  where the expression in the parenthesis satisfies the stated bound.
  
  For $j + 1 < r_L$, the last summand can be written as 
  \begin{equation}\label{eq:last_summand1}
    \frac{1}{U^{\ell}} E_{\ell, j} g' \cdot (f^{(j + 1)} \circ g) = \frac{P_xP_yl_{r_l}(g)}{U^{\ell + 1}}E_{\ell, j} \cdot (f^{(j + 1)} \circ g).
  \end{equation}
  For $j = r_L + 1$, due to Observation~\ref{ob:2}
  \begin{equation}\label{eq:last_summand2}
    g' \cdot (f^{(r_L)} \circ g) = -\frac{P_xP_y}{U} \sum\limits_{j = 0}^{r_L - 1} l_j(g) (f^{(j)} \circ g).
  \end{equation}
  Right-hand sides of both~\eqref{eq:last_summand1} and~\eqref{eq:last_summand2} satisfy the bound.
\end{proof}

Let $f_1, \ldots, f_{r_L}$ be $C$-linearly independent solutions of~$L$, and let $g_1, \ldots, g_{r_P}$ be distinct solutions of $P$.
By $r$ we denote the $C$-dimension of the $C$-linear space $V$ spanned by $f_i \circ g_j$ for all $1 \leq i \leq r_L$ and $1 \leq j \leq r_P$.
The order of the operator annihilating $V$ is at least $r$.
We will construct an operator of order $r$ annihilating $V$ using Wronskian-type matrices.

\begin{lemma}\label{lemma:wronskian}
  There exists a matrix $A(x,y)\in C[x,y]^{(r + 1) \times r_L}$ such that 
  the bidegree of every entry of the $i$-th row of $A(x, y)$ does not exceed $(2rd_P - i + 1, r(2r_P + d_L - 1))$ and 
  $f \in V$ if and only if the vector $(f, \ldots, f^{(r)})^{T}$ lies in 
  the column space of the $(r + 1) \times r_Lr_P$ matrix $\begin{pmatrix} A(x, g_1) & \cdots & A(x, g_{r_P})\end{pmatrix}$.
\end{lemma}

\begin{proof}
  With the notation of Lemma~\ref{lemma:derivative_of_composition}, let 
  $A(x, y)$ be the matrix whose $(i, j)$-th entry is $E_{i - 1, j - 1}(x, y) U(x, y)^{r + 1 - i}$.
  Then $A(x, y)$ meets the stated degree bound.

  By $W_i$ we denote the $(r + 1) \times r_L$ Wronskian matrix for $f_1 \circ g_i, \ldots, f_{r_L} \circ g_i$.
  Then $f \in V$ if and only if the vector $(f, \ldots, f^{(r)})^T$ lies in the column space of the matrix $\begin{pmatrix} W_1 & \cdots & W_{r_{P}} \end{pmatrix}$.
  Hence, it is sufficient to prove that $W_i$ and $A(x, g_i)$ have the same column space.
  The following matrix equality follows from the definition of $E_{i, j}$
  \[
  W_i = \frac{1}{U(x, g_i)^r}A(x, g_i)
  \begin{pmatrix}
    f_1 \circ g_i & \cdots & f_{r_L} \circ g_i \\
    f_1' \circ g_i & \cdots & f_{r_{L}} \circ g_i \\
    \vdots & \ddots & \vdots \\
    f_1^{(r_L - 1)} \circ g_i & \cdots & f_{r_L}^{(r_L - 1)} \circ g_i
  \end{pmatrix}.
  \]
  The latter matrix is nondegenerate since it is a Wronskian matrix for the $C$-linearly independent power series 
  $f_1 \circ g_i$, \dots, $f_{r_L} \circ g_i$ with respect to the derivation $(g_i')^{-1} \partial$.
  Hence, $W_i$ and $A(x, g_i)$ have the same column space.
\end{proof}

In order to express the above condition of lying in the column space in terms of vanishing of a single determinant, we want to ``square'' the matrix $\begin{pmatrix} A(x, g_1), \cdots, A(x, g_{r_P})\end{pmatrix}$.

\begin{lemma}\label{lemma:square}
  There exists a matrix $B(y)\in C[y]^{(r_L r_P - r) \times r_L}$ such that the degree of every entry does
  not exceed $r_P - 1$ and the $(r_Lr_P + 1)\times r_Lr_P$ matrix
  \[
    C =
    \begin{pmatrix}
      A(x, g_1) & \cdots & A(x, g_{r_P}) \\
      B(g_1) & \cdots & B(g_{r_P})    
    \end{pmatrix}
  \]
  has rank $r_Lr_P$.
\end{lemma}

\begin{proof}
  Let $D$ be the Vandermonde matrix for $g_1, \ldots, g_{r_P}$,
  and let $I_{r_L}$ denote the identity matrix.
  Then $C_0 = D \otimes I_{r_L}$ is nondegenerate and has the form $\begin{pmatrix} B_0(g_1), \ldots, B_0(g_{r_P}) \end{pmatrix}$, 
  for some $B_0(y)\in C[y]^{r_Lr_P \times r_L}$ with entries of degree at most $r_P - 1$.
  Since $C_0$ is nondegenerate, we can choose $r_Lr_P - r$ rows which span a complimentary subspace to the row space of $\begin{pmatrix} A(x, g_1), \ldots, A(x, g_{r_P})\end{pmatrix}$.
  Discarding all other rows from $B_0(y)$, we obtain $B(y)$ with the desired properties.
\end{proof}

By $C_\ell$ ($A_\ell(x, y)$, resp.) we will denote the matrix $C$ ($A(x, y)$, resp.) without the $\ell$-th row.

\begin{lemma}\label{lemma:divisibility}
  For every $1 \leq \ell \leq r + 1$ the determinant of $C_\ell$ is divisible by $\prod_{i < j} (g_i - g_j)^{r_L}$
\end{lemma}

\begin{proof}
  We show that $\det C_\ell$ is divisible by $(g_i - g_j)^{r_L}$ for every $i \neq j$.
  Without loss of generality, it is sufficient to show this for $i = 1$ and $j = 2$.
  We have
  \[
    \det C_\ell =
    \begin{vmatrix} 
      \!A_\ell(x, g_1){-}A_\ell(x, g_2)\!&\!A_\ell(x, g_2)\!&\!\cdots\!&\!A_\ell(x, g_{r_P})\!\\
      \!B(g_1){-}B(g_2)\!&\!B(g_2)\!&\!\cdots\!&\!B(g_{r_P})\!
    \end{vmatrix}.
  \]
  Since for every polynomial $p(y)$ we have $g_1 - g_2 \mid p(g_1) - p(g_2)$, every entry of the first $r_L$ columns in the above matrix is divisible by $g_1 - g_2$.
  Hence, the whole determinant is divisible by $(g_1 - g_2)^{r_L}$.
\end{proof}

\begin{theorem}\label{thm:min_operator}
  The minimal operator $M \in C[x][\partial]$ annihilating $f \circ g$ for every $f$ and $g$ such that $L(f) = 0$ and $P(x, g(x)) = 0$
  has order $r \leq r_Lr_P$ and degree at most
  \begin{alignat*}1
    &2r^2d_P - \tfrac12(r{-}2)(r{-}1) + rd_Pr_L (2r_P {+} d_L {-} 1) - d_P r_L (r_P{-}1)\\
    &=\O(r d_Pr_L(d_L+r_P)).
  \end{alignat*}
\end{theorem}

\begin{proof}
  We construct $M$ using $\det C_\ell$ for $1 \leq \ell \leq r + 1$.
  We consider some $f$ and by $F$ we denote the $(r_Lr_P + 1)$-dimensional vector $(f, \ldots, f^{(r)}, 0, \ldots, 0)^T$.
  If $f \in V$, then the first $r + 1$ rows of the matrix $\begin{pmatrix} C & F \end{pmatrix}$ are linearly dependent, so it is degenerate.
  On the other hand, if this matrix is degenerate, then Lemma~\ref{lemma:square} implies that $F$ is a linear combination of the columns of $C$,  so Lemma~\ref{lemma:wronskian} implies that $f \in V$.
  Hence $f \in V \Leftrightarrow \det C_1f \pm \cdots + (-1)^r \det C_{r + 1}f^{(r)} = 0$.
  Due to Lemma~\ref{lemma:divisibility}, the latter condition is equivalent to $c_1f + \cdots + c_{r + 1}f^{(r)} = 0$,
  where $c_\ell= (-1)^{\ell - 1} \det C_\ell/\prod\limits_{i < j} (g_i - g_j)^{r_L}$.
  Thus we can take $M = c_1 + \cdots + c_{r + 1} \partial^{r}$.
  It remains to bound the degrees of the coefficients of~$M$.

  Combining lemmas~\ref{lemma:wronskian}, \ref{lemma:square}, and~\ref{lemma:divisibility}, we obtain
  \begin{align*}
   d_X &:= \deg_x c_\ell \leq \sum\limits_{i \neq \ell} (2rd_P{+}1{-}i) \leq 2r^2d_P - \tfrac12(r{-}2)(r{-}1), \\ 
   d_Y &:= \deg_{g_i} c_\ell \leq rr_L (2r_P + d_L - 1) - r_L(r_P - 1).
  \end{align*}
  Since $c_\ell$ is symmetric with respect to $g_1, \ldots, g_{r_P}$, it can be written as an element of $C[x, s_1, \ldots, s_{r_P}]$ 
  where $s_j$ is the $j$-th elementary symmetric polynomial in $g_1, \ldots, g_{r_P}$, 
  and the total degree of $c_{\ell}$ with respect to $s_j$'s does not exceed~$d_Y$.
  Substituting $s_j$ with the corresponding coefficient of $\frac{1}{\lc_y P}P(x, y)$ and clearing denominators, we obtain a polynomial
  in $x$ of degree at most $d_X + d_Yd_P$.

  Since the order of $M$ is equal to the dimension of the space of all compositions of the form $f \circ g$, 
  where $L(f) = 0$ and $P(x, g) = 0$, $M$ is the minimal annihilating operator for this space.
\end{proof}

\begin{remark}
  The proof of Theorem~\ref{thm:min_operator} is a generalization of the proof of \cite[Thm.~1]{bostan07}.
  Specializing $r_L = 1$, $d_L = 0$ in Theorem~\ref{thm:min_operator} gives a sightly larger bound as the
  bound in~\cite[Thm.~1]{bostan07}, but with the same leading term. 
\end{remark}

Although the bound of Theorem~\ref{thm:min_operator} for $r=r_Lr_P$ beats the bound of Theorem~\ref{thm:linalgcurve}
for $r=r_Lr_P$ by a factor of~$r_P$, it is apparently still not tight. Experiments
we have conducted with random operators lead us to conjecture that in fact, at least
generically, the minimal order operator of order $r_Lr_P$ has degree $\O(r_L r_P d_P  (d_L + r_L r_P))$.
By interpolating the degrees of the operators we found in our computations,
we obtain the expression in the following conjecture.

\begin{conjecture}\label{conj}
  For every $r_P,r_L,d_P,d_L\geq2$ there exist $L$ and $P$ such that the corresponding
  minimal order operator $M$ has order $r_Lr_P$ and degree
  \begin{alignat*}1 
    & r_L^2 (2 r_P(r_P-1) + 1) d_P 
    +r_L r_P (d_P(d_L+1) + 1)
    +d_L d_P\\
    &\qquad{}    
    -r_L^2 r_P^2
    -r_L d_L d_P,
  \end{alignat*}
  and there do not exist $L$ and $P$ for which the corresponding minimal operator $M$
  has order $r_Lr_P$ and larger degree. 
\end{conjecture}


\section{Order-Degree-Curve\hfill\break by singularities}\label{sec:3}

A singularity of the minimal operator $M$ is a root of its leading coefficient polynomial $\lc_\partial(M)\in C[x]$.
In the notation and terminology of~\cite{jaroschek13a}, a factor~$p$ of this polynomial is called \emph{removable} at cost~$n$ if there exists an operator
$Q\in C(x)[\partial]$ of order $\deg_\partial(Q)\leq n$ such that $QM\in C[x][\partial]$ and $\gcd(\lc_\partial(QM),p)=1$.
A factor~$p$ is called \emph{removable} if it is removable at some finite cost $n\in\set N$, and \emph{non-removable} otherwise. 
The following theorem~\cite[Theorem~9]{jaroschek13a} translates information about the removable singularities of a minimal operator
into an order-degree curve.

\begin{theorem}\label{thm:9} Let $M\in C[x][\partial]$, and let $p_1,\dots,p_m\in C[x]$ be pairwise coprime
  factors of $\lc_\partial(M)$ which are removable at costs $c_1,\dots,c_m$, respectively. Let $r\geq\deg_\partial(M)$
  and
  \[
    d\geq\deg_x(M)-\biggl\lceil\sum_{i=1}^m\Big(1-\frac{c_i}{r-\deg_\partial(M)+1}\Bigr)^+\deg_x(p_i)\biggr\rceil,
  \]
  where we use the notation $(x)^+:=\max\{x,0\}$. Then there exists an operator $Q\in C(x)[\partial]$ such that
  $QM\in C[x][\partial]$ and $\deg_\partial(QM)=r$ and $\deg_x(QM)=d$. 
\end{theorem}

The order-degree curve of Theorem~\ref{thm:9} is much more accurate than that of Theorem~\ref{thm:linalgcurve}.
However, the theorem depends on quantities that are not easily observable when only $L$ and $P$ are known.
{}From Theorem~\ref{thm:min_operator} (or Conj.~\ref{conj}), we have a good bound for $\deg_x(M)$. In the
the rest of the paper, we discuss bounds and plausible hypotheses for the degree and the cost of the removable factors.
The following example shows how knowledge about the degree of the operator and the degree and cost of its removable
singularities influence the curve. 

\begin{example}
  The figure below compares the data of Example~\ref{ex:1} with the curve obtained from Theorem~\ref{thm:9}
  using $m=1$, $\deg_x(M_{\min})=544$, $\deg_x(p_1)=456$, $c_1=1$.
  This curve is labeled (a) below. Only for a few orders~$r$, the curve slightly overshoots.
  In contrast, the curve of Theorem~\ref{thm:linalgcurve}, labeled (b) below, overshoots significantly
  and systematically.
  
  The figure also illustrates how the parameters affect the accuracy of the estimate.
  The value $\deg_x(M_{\min})=544$ is correctly predicted by Conjecture~\ref{conj}. If we
  use the more conservative estimate $\deg_x(M_{\min})=1568$ of Theorem~\ref{thm:min_operator},
  we get the curve~(e). 
  For curve~(d) we have assumed a removability degree of $\deg_x(p_1)=408$, as predicted by Theorem~\ref{thm:curve}
  below, instead of the true value~$\deg_x(p_1)=456$.
  For (c) we have assumed a removability cost $c_1=10$ instead of $c_1=1$.
  
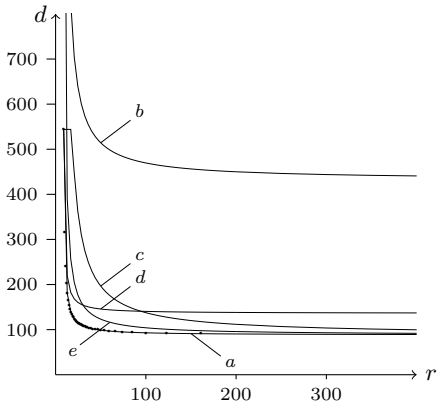
\begin{center}
    \begin{tikzpicture}[yscale=.6,xscale=1.2,scale=1]
      \draw[->] (0,0)--(0,8) node[left] {$d$};
      \draw[->] (0,0)--(4,0) node[right] {$r$};
      \foreach\x in {1,...,7} \draw (0,\x)--(-.1,\x) node[left] {\scriptsize$\x 00$};
      \foreach\x in {1,...,3} \draw (\x,0)--(\x,-.1) node[below] {\scriptsize$\x 00$};
      \clip (-1,-1) rectangle (5,8);
      \begin{scope}[ultra thin, variable=\r,domain=.09:2,samples=50]
        \draw plot ({\r},{432*\r/(100*\r-8)});
        \draw plot ({\r},{.08*(-31+1100*\r)/(100*\r-8)});
        \draw plot ({\r},{min(5.44,.08*(482+1100*\r)/(100*\r-8))});
        \draw plot ({\r},{1.36*(-5+100*\r)/(100*\r-8)});
        \draw plot ({\r},{8*(.97+11.00*\r)/(100*\r-8)});
      \end{scope}
      \begin{scope}[ultra thin, variable=\r,domain=2:4,samples=5]
        \draw plot ({\r},{432*\r/(100*\r-8)});
        \draw plot ({\r},{.08*(-31+1100*\r)/(100*\r-8)});
        \draw plot ({\r},{.08*(482+1100*\r)/(100*\r-8)});
        \draw plot ({\r},{1.36*(-5+100*\r)/(100*\r-8)});
        \draw plot ({\r},{8*(.97+11.00*\r)/(100*\r-8)});
      \end{scope}
      \draw[ultra thin] (.50,5.14286) -- ++(.33,.5) node[above right,xshift=-2pt,yshift=-2pt] {\scriptsize b};
      \draw[ultra thin] (.50,1.96571) -- ++(.33,.5) node[above right,xshift=-2pt,yshift=-2pt] {\scriptsize c};
      \draw[ultra thin] (.50,1.45714) -- ++(.33,.5) node[above right,xshift=-2pt,yshift=-2pt] {\scriptsize d};
      \draw[ultra thin] (.60,1.16077) -- ++(-.33,-.5) node[below left,xshift=2pt,yshift=2pt] {\scriptsize e};
      \draw[ultra thin] (1.5,.912113) -- ++(.33,-.5) node[below right,xshift=-2pt,yshift=2pt] {\scriptsize a};
      \draw
(.09,5.44) node {$\cdot$}
(.10,3.16) node {$\cdot$}
(.11,2.40) node {$\cdot$}
(.12,2.02) node {$\cdot$}
(.13,1.79) node {$\cdot$}
(.14,1.64) node {$\cdot$}
(.15,1.53) node {$\cdot$}
(.16,1.45) node {$\cdot$}
(.17,1.38) node {$\cdot$}
(.18,1.33) node {$\cdot$}
(.19,1.29) node {$\cdot$}
(.20,1.26) node {$\cdot$}
(.21,1.23) node {$\cdot$}
(.22,1.20) node {$\cdot$}
(.23,1.18) node {$\cdot$}
(.24,1.16) node {$\cdot$}
(.25,1.14) node {$\cdot$}
(.26,1.13) node {$\cdot$}
(.27,1.12) node {$\cdot$}
(.28,1.10) node {$\cdot$}
(.29,1.09) node {$\cdot$}
(.30,1.08) node {$\cdot$}
(.31,1.07) node {$\cdot$}
(.33,1.06) node {$\cdot$}
(.34,1.05) node {$\cdot$}
(.35,1.04) node {$\cdot$}
(.37,1.03) node {$\cdot$}
(.39,1.02) node {$\cdot$}
(.41,1.01) node {$\cdot$}
(.44,1.00) node {$\cdot$}
(.47,.99) node {$\cdot$}
(.50,.98) node {$\cdot$}
(.54,.97) node{$\cdot$}
(.59,.96) node{$\cdot$}
      (.66,.95) node{$\cdot$}
      (.74,.94) node{$\cdot$}      
      (.85,.93) node{$\cdot$}   
      (1.00,.92) node{$\cdot$}   
      (1.23,.91) node{$\cdot$}
      (1.61,.90) node{$\cdot$}
      ;
    \end{tikzpicture}
  \end{center}
\end{example}


\subsection{Degree of Removable Factors}

\begin{lemma}\label{lemma:multiplicity_resultant}
  Let $P(x, y)\in C[x, y]$ be a polynomial with $\deg_y P = d$, and $R(x) = \Res_y(P, P_y)$.
  Assume that $\alpha\in\bar C$ is a root of $R(x)$ of multiplicity $k$.
  Then the squarefree part
  \[
  S(y) = P(\alpha, y)\bigm/\gcd\bigl( P(\alpha, y), P_y(\alpha, y) \bigr)
  \]
  of $P(\alpha, y)$ has degree at least $d - k$.
\end{lemma}

\begin{proof}
  Let $M(x)$ be the Sylvester matrix for $P(x, y)$ and $P_y(x, y)$ with respect to $y$.
  The value $R^{(k)}(\alpha)$ is of the form $\sum \det M_i(\alpha)$, where every $M_i(x)$ has
  at least $2d - 1 - k$ common columns with $M(x)$. Since $R^{(k)}(\alpha) \neq 0$, at least one of 
  these matrices is nondegenerate. Hence, $\corank M(\alpha) \leq k$.
  On the other hand, $\corank M(\alpha)$ is equal to the dimension of the space of pairs of polynomials
  $(a(y), b(y))$ such that $a(y)P(\alpha, y) + b(y)P_y(\alpha, y) = 0$ and $\deg b(y) < d$.
  Then $b(y)$ is divisible by $S(y)$, and for every $b(y)$ divisible by $S(y)$ there exists exactly one~$a(y)$.
  Hence, $\corank M(\alpha) = d - \deg S(y) \leq k$.
\end{proof}

Let $M$ be the minimal order operator annihilating all compositions $f\circ g$ of a solution of $P$ with a solution of~$L$.
The leading coefficient $q = \lc_{\partial}(M)\in C[x]$ can be factored as $q = q_{\rem}q_{\nonrem}$, 
where $q_{\rem}$ and $q_{\nonrem}$ are the products of all removable and all nonremovable factors of $\lc_{\partial}(M)$, respectively.

\begin{lemma}
  $\deg q_{\nonrem} \leq d_P(4r_Lr_P - 2r_L + d_L)$.
\end{lemma}

\begin{proof}
  For $\alpha \in\bar C$ by $\pi_{\alpha}$ ($\lambda_{\alpha}$, $\mu_{\alpha}$, resp.) we denote $r_P$
  ($r_L$ or $\deg_{\partial} M$, resp.) 
  minus the number of solutions of $P(x, g(x)) = 0$ (the dimension of the solutions set of $Lf(x) = 0$ or $Mf(x) = 0$, resp.)
  in $\bar C[[x - \alpha]]$.

  Corollary~4.3 from~\cite{tsai00} implies that $\ord_\alpha q_{\nonrem}$ (the minimal order at $\alpha$ in $\In_\alpha\Cl_\alpha(M)$ in notation of~\cite{tsai00})
  is equal to $\mu_\alpha$ ($\ord_\alpha B_\alpha(M) - (s_\alpha + 1)$ in notation of~\cite{tsai00}).
  Summing over all $\alpha$, we have $\sum\limits_{\alpha\in\bar C} \mu_\alpha = \deg q_{\nonrem}$.
  Bounding the degree of the nonremovable part of $\lc_\partial (L)$ by $d_L$, we also have $\sum\limits_{\alpha \in\bar C} \lambda_\alpha \leq d_L$.
  
  Let $R(x)$ be the resultant of $P(x, y)$ and $P_y(x, y)$ with respect to $y$.
  Let $\alpha$ be a root of $R(x)$ of multiplicity $k$.
  Lemma~\ref{lemma:multiplicity_resultant} implies that the degree of the squarefree part of $P(\alpha, y)$ is at least $r_P - k$.
  So, at most $k$ roots are multiple, so at least $r_P - 2k$ roots are simple.
  Hence, $P(x, y) = 0$ has at least $r_P - 2k$ solutions in $\bar C[[x - \alpha]]$.
  Thus $\sum_{\alpha\in\bar C} \pi_\alpha \leq 2\deg R \leq 2d_P(2r_P - 1)$.

  Let $\alpha \in\bar C$ and let $g_1(x), \ldots, g_{r_P - \pi_\alpha}(x) \in\bar C[[x - \alpha]]$ 
  be solutions of $P(x, g(x)) = 0$.
  Let $\beta_i = g_i(0)$ for all $1 \leq i \leq r_P - \pi_\alpha$.
  Since the composition of a power series in $x - \beta_i$ with $g_i(x)$ is a power series in $x - \alpha$,
  \begin{equation}\label{ineq:mu_bound}
  \mu_\alpha \leq r_L \pi_\alpha + \sum\limits_{i = 1}^{r_P - \pi_\alpha} \lambda_{\beta_i}.
  \end{equation}
  We sum~\eqref{ineq:mu_bound} over all $\alpha \in\bar C$.
  The number of occurrences of $\lambda_\beta$ in this sum for a fixed $\beta \in\bar C$ is equal to the number
  of distinct power series of the form
  $g(x) = \beta + \sum c_i (x - \gamma)^i$ such that $P(x, g(x)) = 0$.
  Inverting these power series, we obtain distinct Puiseux series solutions of $P(x, y) = 0$ at $y = \beta$, so this number does not exceed $d_P$.
  Hence
  \[
    \sum\limits_{\alpha \in\bar C} \mu_\alpha \leq r_L \sum\limits_{\alpha \in\bar C} \pi_\alpha + d_P \sum\limits_{\beta \in\bar C} \lambda_\beta \leq 2r_L d_P(2r_P - 1) + d_Pd_L.\kern-1ex\qedhere
  \]
\end{proof}

In order to use Theorem~\ref{thm:9}, we need a lower bound for $\deg q_{\rem}$.
Theorem~\ref{thm:min_operator} gives us an upper bound for $\deg_x M$, but we must also estimate the difference $\deg_x M - \deg \lc_\partial M$.
By $N_\alpha$ we denote the Newton polygon for $M$ at $\alpha\in\bar C\cup\{\infty\}$
(for definitions and notation, see~\cite[Section~3.3]{vanhoeij97a}).
By $H_\alpha$, we denote the difference of the ordinates of the highest and the smallest vertices of~$N_\alpha$,
and we call this quantity the \emph{height} of the Newton polygon.
Note that $H_\infty\leq\deg_x M - \deg \lc_\partial M$.
This estimate together with the Lemma above implies
$\deg q_{\rem}\geq\deg_x(M)-H_\infty-d_P(4r_Lr_P - 2r_L + d_L)$. 

The equation $P(x, y) = 0$ has $r_P$ distinct Puiseux series solutions $g_1(x), \ldots, g_{r_P}(x)$ at infinity.
For $1 \leq i \leq r_P$, let $\beta_i = g_i(\infty) \in\bar C\cup\{\infty\}$, and let $\rho_i$ be the order of
zero of $g_i(x) - \beta_i$ ($\frac{1}{g_i(x)}$, resp.) at infinity if $\beta_i \in\bar C$ ($\beta_i  = \infty$, resp.).
The numbers $\rho_1, \ldots, \rho_{r_P}$ are positive rationals and can be read off from Newton polygons of~$P$ (see~\cite[Chapter II]{bliss08}).

\begin{lemma}
  $H_\infty\leq\sum_{i=1}^{r_P}\rho_i H_{\beta_i}$.
\end{lemma}

\begin{proof}
  Writing $L$ as $L(x, \partial) \in C[x][\partial]$, we have
  $$
    M = \LCLM \left( L\left(g_1, \frac{1}{g_1'} \partial\right), \ldots, L\left( g_{r_P}, \frac{1}{g_{r_P}'} \partial\right) \right).
  $$
    Hence, the set of edges of $N_\infty$ is a subset of the union of sets of edges of Newton polygons of the operators $L(g_i,\frac1{g_i'}\partial)$,
  so the height of $N_\infty$ is bounded by the sum of the heights of the Newton polygons of these operators.
  Consider $g_1$ and assume that $\beta_1 \in\bar C$.
  Then the Newton polygon for $L$ at $\beta_1$ is constructed from the set of monomials of $L$ written as an element of $C(x - \beta_1)[(x - \beta_1)\partial]$.
  Let $L(x, \partial) = \tilde{L}(x - \beta_1, (x- \beta_1)\partial)$, then
  \[
  L\bigl(g_1, \frac{1}{g_1'} \partial\bigr)
  = \tilde{L}\bigl( g_1 - \beta_1, \frac{g_1{-}\beta_1}{g_1'}\partial \bigr)
  = \tilde{L}\bigl( x^{-\rho_1}h_1(x), xh_2(x) \partial\bigr),
  \]
  where $h_1(\infty)$ and $h_2(\infty)$ are nonzero elements of $\bar{C}$.
  Since $h_1$ and $h_2$ do not affect the shape of the Newton polygon at infinity, the Newton polygon at infinity for $L(g_1, \frac{1}{g_1'} \partial)$ 
  is obtained from the Newton polygon for $L$ at $\beta_1$ by stretching it vertically by the factor~$\rho_1$,
  so its height is equal to $\rho_1H_{\beta_1}$.

  The case $\beta_1 = \infty$ is analogous using $L = \tilde{L}\left( \frac{1}{x}, -x\partial\right)$.
\end{proof}

\begin{remark}\label{rem:newton}
  Generically, the $\beta_i$'s will be ordinary points of~$L$, so it is
  fair to expect $H_{\beta_i}=0$ for all $i$ in most situations.
\end{remark}

The following theorem is a consequence of Theorem~\ref{thm:9} and the discussion above.

\begin{theorem}\label{thm:curve}
  Let $\rho_1, \ldots, \rho_{r_P}$ be as above.
  Assume that all removable singularities of $M$ are removable at cost at most~$c$.
  Let $\delta = \sum\limits_{i = 1}^{r_P} \rho_i H_{\beta_i} + d_P(4r_Lr_P - 2r_L + d_L)$.
  Let $r \geq \deg_{\partial} M + c - 1$ and
  \[
  d \geq \delta \cdot \Bigl(1-\frac{c}{r-\deg_\partial(M)+1}\Bigr) + \deg_x M \cdot \frac{c}{r - \deg_\partial (M) + 1}.
  \]
  Then there exists an operator $Q \in C(x)[\partial]$ such that $QM \in C[x][\partial]$ and $\deg_{\partial} (QM) = r$ and $\deg_x (QM) = d$.
\end{theorem}

Note that $\deg_x(M)$ may be replaced with the expression from Theorem~\ref{thm:min_operator}
or Conjecture~\ref{conj}.


\subsection{Cost of Removable Factors}

The goal of this final section is to explain why in the case $r_P > 1$ one can almost always choose $c=1$ in Theorem~\ref{thm:curve}.

For a differential operator $L \in C[x][\partial]$, by $M(L)$ we denote the minimal operator $M$ such that $Mf(g(x)) = 0$ whenever $L f = 0$ and $P(x, g(x)) = 0$.
We want to investigate the possible behaviour of a removable singularity at $\alpha \in C$ when $L$ varies and $P$ with $r_P>1$ is fixed.
Without loss of generality, we assume that $\alpha = 0$.

We will assume that:
\begin{enumerate}
  \item[(S1)]\label{s1} $P(0, y)$ is a squarefree polynomial of degree~$r_P$;

  \kern-\smallskipamount
  \item[(S2)]\label{s2} $g(0)$ is not a singularity of $L$ for any root $g(x)$ of $P$;

  \kern-\smallskipamount
  \item[(G)]\label{g} Roots of $P(x, g(x)) = 0$ at zero are of the form $g_i(x) = \alpha_i + \beta_i x + \gamma_i x^2 + \ldots$, where $\beta_2, \ldots, \beta_{r_P}$ are nonzero, 
  and either $\beta_1$ or $\gamma_1$ is nonzero.
\end{enumerate}

Conditions \hyperref[s1]{(S1)} and \hyperref[s2]{(S2)} ensure that zero is not a potential true singularity of $M(L)$.
Condition \hyperref[g]{(G)} is an essential technical assumption on~$P$.
We note that it holds at all nonsingular points (not just at zero) for almost all $P$, because this condition is violated at $\alpha$ iff some root of $P(\alpha, y) = P_x(\alpha, y) = 0$ (this means that at least one of $\beta_i$ is zero) is also a root of either $P_{xx}(\alpha, y) = 0$ (then $\gamma_i$ is also zero) or $P_{xy}(\alpha, y) = 0$ (then there are at least two such $\beta$'s).
For a generic $P$ this does not hold. 

Under these assumptions we will prove the following theorem.
Informally speaking, it means that if $M(L)$ has an apparent singularity at zero, then it almost surely is removable at cost one.

\begin{theorem}\label{thm:generic}
  Let $d_L$ be such that $d_L \geq (r_Lr_P - r_L + 1)r_P$.
  By $V$ we denote the (algebraic) set of all $L \in \bar{C}[x][\partial]$ of order $r_L$ and degree $\leq d_L$ 
  such that the leading coefficient of $L$ does not vanish at $\alpha_1, \ldots, \alpha_{r_P}$.
  We consider two (algebraic) subsets in $V$
  \begin{alignat*}1
          X &= \bigl\{ L \in V \bigm| \text{$M(L)$ has an apparent singularity at $0$}\bigr\},\\
          Y &= \bigl\{ L \in V \bigm| \text{$M(L)$ has an apparent singularity at $0$}\\[-2pt]
            &\hspace{50pt} \text{which is not removable at cost one} \bigr\}.
   \end{alignat*}
   Then, $\dim X > \dim Y$ as algebraic sets. 
\end{theorem}
   
For $\alpha \in \bar{C}$, by $\Op_{\alpha}(r, d)$ we denote the space of differential operators in $\bar{C}[x - \alpha][\partial]$ of order at most $r$ and degree at most~$d$.
By $\NOp_{\alpha}(r, d) \subset \Op_{\alpha}(r, d)$ we denote the set of $L$ such that $\ord L = r$ and $(\lc_{\partial} L) (\alpha) \neq 0$. 
Then $$V \subset \NOp_{\alpha_1}(r_L, d_L) \cap \ldots \cap \NOp_{\alpha_{r_P}} (r_L, d_L).$$

To every operator $L \in \NOp_{\alpha}(r, d_0)$ and $d_1 \geq r$, we assign \emph{a fundamental matrix of degree $d_1$} at $\alpha$, denote it by $F_{\alpha}(L, d_1)$.
It is defined as the $r \times (d_1 + 1)$ matrix such that the first $r$ columns constitute the identity matrix~$I_r$, 
and every row consists of the first $d_1 + 1$ terms of some power series solution of $L$ at $x = \alpha$.
Since $L \in \NOp_{\alpha}(r, d_0)$, $F(L, d_1)$ is well defined for every $d_1$.

By $F(r, d)$ we denote the space of all possible fundamental matrices of degree $d$ for operators of order $r$.
This space is isomorphic to $\mathbb{A}^{r(d + 1 - r)}$.
The following proposition says that a generic operator has generic and independent fundamental matrices, so we can work with these matrices instead of working with operators.

\begin{proposition}\label{prop:generic_solutions}
  Let $\varphi\colon V \to \left( F(r_L, r_Lr_P) \right)^{r_P}$ be the map sending $L \in V$ to $F_{\alpha_1}(L, r_Lr_P) \oplus \ldots \oplus F_{\alpha_{r_P}}(L, r_Lr_P)$.
  Then $\varphi$ is a surjective map of algebraic sets, and all fibers of $\varphi$ have the same dimension.
\end{proposition}

For the proof we need the following lemma.
\begin{lemma}\label{lem:generic_solutions}
  Let $\psi \colon \NOp_{\alpha}(r, d) \to F(r, d + r)$ be the map sending $L$ to $F_{\alpha}(L, d + r)$.
  Then $\psi$ is surjective and all fibers have the same dimension.
\end{lemma}

\begin{proof}
  First we assume that $L$ is of the form $L = \partial^{r_L} + a_{r_L - 1}(x) \partial^{r_L - 1} + \ldots + a_0(x)$, 
  and $a_j(x) = a_{j, d}x^d + \ldots + a_{j, 0}$, where $a_{j, i} \in \bar{C}$.
  We also denote the truncated power series corresponding to the $j + 1$-st row of $F(L, d + r_L)$ by $f_j$ and write it as
  $$
  f_j = x^{j} + \sum\limits_{i = 0}^d b_{j, i} x^{r_L + i}, \text{ where } b_{j, i} \in \bar{C}.
  $$
  We will prove the following claim by induction on $i$:
          
  \textbf{Claim.} \textit{For every $0 \leq j \leq r_L - 1$ and every $0 \leq i \leq d$, $b_{j, i}$ can be written as a polynomial in $a_{p, q}$ with $q < i$ and $a_{j, i}$.
  And, vice versa, $a_{j, i}$ can be written as a polynomial in $b_{p, q}$ with $q < i$ and $b_{j, i}$.}
  
  The claim would imply that $\psi$ defines an isomorphism of algebraic varieties between $F_{\alpha}(r_P, d + r)$ and the subset of monic operators in $\NOp_{\alpha}(r, d)$.
      
  For $i = 0$, looking at the constant term of $L (f_j)$, we obtain that $j! a_{j, 0} + r_L! b_{j, 0} = 0$.
  This proves the base case of the induction.
  
  Now we consider $i > 0$ and look at the constant term of $\partial^i L(f_j)$.
  The operator $\partial^i L$ can be written as
  \begin{alignat*}1
    \partial^i L &=
    \partial^{i + r_L} + a_{r_L - 1}^{(i)}(x) \partial^{r_L - 1} + \ldots + a_0^{(i)}(x)\\
    &\qquad{}+ \sum\limits_{k < i, l < i + r_L, s \leq d} c_{k, l, s} a_s^{(k)}(x) \partial^l 
  \end{alignat*}
  Applying this to $f_j$, we obtain the following expression for the constant term:
  $$
  (i + r_L)! b_{j, i} + j! i! a_{j, i} + \sum\limits_{k < i, l < i + r_L, s \leq d} \tilde{c}_{k, l, s} a_{s, k} b_{j, l - r_L} = 0.
  $$
  Applying the induction hypothesis to the equalities
  $$
  b_{j, i} = \frac{-1}{(i + r_L)!} \left( j! i! a_{j, i} + \sum\limits_{k < i, l < i + r_L, s \leq d} \tilde{c}_{k, l, s} a_{s, k} b_{j, l - r_L} \right)
  $$
  $$
  a_{j, i} = \frac{-1}{i! j!} \left( (i + r_L)! b_{j, i} + \sum\limits_{k < i, l < i + r_L, s \leq d} \tilde{c}_{k, l, s} a_{s, k} b_{j, l - r_L} \right)
  $$
  we prove the claim.

  The above proof also implies that $F(L, d + r)$ is completely determined by the truncation of $L$ at degree $d + 1$.
  So, for arbitrary $L \in \NOp_{\alpha}(r, d)$, $F(L, d) = F(\tilde{L}, d)$, where $\tilde{L}$ is the truncation of $\frac{1}{\lc_\partial L} L$ at degree $d + 1$,
  which is monic in~$\partial$.
  Hence, every fiber of $\psi$ is isomorphic to the set of all polynomials of degree at most $d$ with nonzero constant term.
  This set is isomorphic to $\bar{C}^{*} \times \bar{C}^d$.
\end{proof}

\begin{proof}[Proof of Proposition~\ref{prop:generic_solutions}]
  Let $d_0=r_Lr_P - r_L$.
  We will factor $\varphi$ as a composition 
  \[
    V \xrightarrow{\varphi_1} \bigoplus_{i=1}^{r_P} \NOp_{\alpha_i}(r_L, d_0)
      \xrightarrow{\varphi_2} F(r_L, r_Lr_P)^{r_P},
  \]
  where $\varphi_2$ is a component-wise application of $F_{\alpha_i}(\ast, d_0)$ and $\varphi_1$ sends $L \in V$ to 
  a vector whose $i$-th coordinate is the truncation at degree $d_0 + 1$ of $L$ written as an element of $\bar{C}[x - \alpha_i][\partial]$.
  We will prove that both these maps are surjective with fibers of equal dimension.

  The map $\varphi_1$ can be extended to $$\varphi_1\colon \Op_0(r_L, d_L) \to \Op_{\alpha_1}(r_L, d_0) \oplus \ldots \oplus \Op_{\alpha_{r_P}}(r_L, d_0).$$
  This map is linear, so it is sufficient to show that the dimension of the kernel is equal to the difference of the dimensions of the source space and the target space.
  The latter number is equal to $(d_L + 1)(r_L + 1) - (d_0 + 1)(r_L + 1)r_P$.
  Let $L \in \Ker\varphi_1$.
  This is equivalent to the fact that every coefficient of $L$ is divisible by $(x - \alpha_i)^{d_0 + 1}$ for every $1 \leq i \leq r_P$.
  The dimension of the space of such operators is equal to $(r_L + 1)(d_L + 1 - r_P(d_0 + 1)) \geq 0$, so $\varphi_1$ is surjective.

  Lemma~\ref{lem:generic_solutions} implies that $\varphi_2$ is also surjective and all fibers are of the same dimension.
\end{proof}


Let $g_1(x), \ldots, g_{r_P}(x) \in \bar{C}[[x]]$ be solutions of $P(x, y) = 0$ at zero.   
Recall that $g_i(x) = \alpha_i + \beta_i x + \ldots$ for all $1 \leq i \leq r_P$, and by~\hyperref[g]{(G)} we can assume that $\beta_2, \ldots, \beta_{r_P}$ are nonzero.

Consider $A \in F(r_L, d)$, assume that its rows correspond to truncations of power series $f_1, \ldots, f_{r_L} \in \bar{C}[[x - \alpha_i]]$.
By $\varepsilon(g_i, A)$ we denote the $r_L \times (d + 1)$-matrix whose rows are truncations of $f_1\circ g_i, \ldots, f_{r_L} \circ g_i \in \bar{C}[[x]]$ at degree $d + 1$.

\begin{lemma}\label{lem:varepsilon}
  We can write $\varepsilon(g_i, A) = A \cdot T(g_i)$,
  where $T(g_i)$ is an upper triangular $(d + 1) \times (d + 1)$-matrix depending only on $g_i$
  with $1, \beta_i, \ldots, \beta_i^d$ on the diagonal.

  Futhermore, if $\beta_i = 0$ and $g_i(x) = \alpha_i + \gamma_i x^2 + \ldots$, then the $i$-th row of $T(g_i)$
  is zero for $i \geq \frac{d + 3}{2}$,
  and starts with $2(i - 1)$ zeroes and $\gamma_i^{i - 1}$ for $i<\frac{d+3}2$.
\end{lemma}

\begin{proof}
  Let the $j$-th row of $A$ correspond to a polynomial $f_j(x - \alpha_i) = x^{j - 1} + O(x^{r_L})$.
  The substitution operation $f_j \to f_j \circ g_i$ is linear with respect to coefficients of $f_i$, so $\varepsilon(g_i, A) = A \cdot T(g_i)$ for some matrix $T(g_i)$.
  Since the coefficient of $x^k$ in $f_j \circ g_i$ is a linear combination of coefficients of $(x - \alpha_i)^l$ with $l \leq k$ in $f_j$, the matrix $T(g_i)$ is upper triangular.
  Since $(x - \alpha_i)^k \circ g_i = \beta_i^k x^k + O(x^{k + 1})$, $T(g_i)$ has $1, \beta_i, \ldots, \beta_i^d$ on the diagonal.

  The second claim of the lemma can be verified by a similar computation.
\end{proof}

\begin{corollary}\label{cor:varepsilon}
  If $\beta_i \neq 0$, then the matrix $\varepsilon(g_i, A)$ has the form $(A_0\; A_1)$, where $A_0$ is an upper triangular matrix over $\bar{C}$, and the entries of $A_1$ are linearly independent linear forms in the entries of~$A$.
\end{corollary}

An element of the affine space $W = \left( F(r_L, r_L r_P)\right)^{r_P}$ is a tuple of matrices $N_1, \ldots, N_{r_P} \in F(r_L, r_Lr_P)$, where
every $N_i$ has the form $N_i = (E_{r_L} \; \tilde{N}_i)$.
Entries of $\tilde{N}_1, \ldots, \tilde{N}_{r_P}$ are coordinates on $W$, so we will view entries of $\tilde{N}_i$ as a set $X_i$ of algebraically independent variables.
We will represent $N$ as a single $(r_L r_P) \times (r_L r_P + 1)$-matrix
$$
N = \begin{pmatrix} N_1 \\ \vdots \\ N_{r_P} \end{pmatrix},\text{ and set } \varepsilon(N) = \begin{pmatrix} \varepsilon(g_1, N_1) \\ \vdots \\ \varepsilon(g_{r_P}, N_{r_P}) \end{pmatrix}.
$$

For any matrix $A$, by $A_{(1)}$ and $A_{(2)}$ we denote $A$ without the last column and without the last but one column, respectively.
By $\pi$ we denote the composition $\varepsilon \circ \varphi$.
Since $\pi(L)$ represents solutions of $M(L)$ at zero truncated at degree $r_Lr_P + 1$, properties of the operator $L \in V$ can be described in terms of the matrix $\pi(L)$:
\begin{itemize}
  \item $M(L)$ has order less than $r_L r_P$ or has an apparent singularity at zero iff $\pi(L)_{(1)}$ is degenerate;
  \item $M(L)$ has order less than $r_L r_P$ or has an apparent singularity at zero which is either not removable at cost one or of degree greater than one iff both $\pi(L)_{(1)}$ and $\pi(L)_{(2)}$ are degenerate.
\end{itemize}
Let $X_0 = \{ L \in V \mid \det\pi(L)_{(1)} = 0\}$ and $Y_0 = \{ L \in V \mid \det\pi(L)_{(2)} = 0\}$,
then $X_0 \setminus Y_0 \subset X \subset X_0$ and $Y \subset Y_0$.

\begin{proposition}\label{prop:irreducible}
  $\varphi(X_0)$ is an irreducible subset of $W$, and $\varphi(Y_0)$ is a proper algebraic subset of $\varphi(X_0)$.
\end{proposition}

\begin{proof}
  The above discussion and the surjectivity of $\varphi$ imply that $\varphi(X_0) = \{ N \in W \mid \det\varepsilon(N)_{(1)} = 0\}$.
  Hence, we need to prove that $\det\varepsilon(N)_{(1)}$ is a nonzero irreducible polynomial in $R = \bar{C}[ X_1, \ldots, X_{r_P} ]$.
  We set $A = \varepsilon(N)_{(1)}$.
  
  We claim that there is a way to reorder columns and rows of $A$ such that it will be of the form
  $$
  \begin{pmatrix}
    B & C_1 \\
    C_2 & D
  \end{pmatrix},
  $$
  where $B$ and $D$ are square matrices, and
  \begin{itemize}
    \item $B$ is upper triangular with nonzero elements of $\bar{C}$ on the diagonal;
    \item entries of $D$ are algebraically independent over the subalgebra generated in $R$ by entries of $B, C_1$, and $C_2$.
  \end{itemize}
  In order to prove the claim we consider two cases:
  \begin{enumerate}
    \item $\beta_1 \neq 0$. By Corollary~\ref{cor:varepsilon}, $A$ is already of the desired form with $B$ being an $r_L \times r_L$-submatrix.
    \item $\beta_1 = 0$. Then~\hyperref[g]{(G)} implies that $g_1(x) = \alpha_1 + \gamma_1 x^2 + \ldots$ with $\gamma_1 \neq 0$.
    Then Lemma~\ref{lem:varepsilon} implies that the following permutations would give us the desired block structure with $B$ being an $\lfloor 3r_L / 2 \rfloor \times \lfloor 3r_L / 2 \rfloor$-submatrix,
    for columns:
    \[
    1, 3, \ldots, 2r_L - 1, 2, 4, \ldots, 2\lfloor r_L / 2\rfloor, \ast,
    \]
    and for rows:
    \[
    1, 2, \ldots, r_L,r_L + 2, r_L + 4, \ldots, r_L + 2\lfloor r_L / 2 \rfloor,\ast,
    \]
    where $\ast$ stands for all other indices in any order.
  \end{enumerate}
  Using elementary row operations, we can bring $A$ to the form
  $$
  \begin{pmatrix}
    B & \ast \\
    0 & \widetilde{D}
  \end{pmatrix},
  $$
  where the entries of $\widetilde D$ are still algebraically independent.
  Hence, $\det A$ is proportional to $\det \widetilde{D}$ which is irreducible.

  In order to prove that $\varphi(Y_0)$ is a proper subset of $\varphi(X_0)$ it is sufficient to prove that $\det \varepsilon(N)_{(2)}$ is not divisible by $\det \varepsilon(N)_{(1)}$.
  This follows from the fact that these polynomials are both of degree $r_Lr_P - r_L$ with respect to (algebraically independent) 
  entries of $\tilde{N}_2, \ldots, \tilde{N}_{r_P}$, but involve different subsets of this variable set.
\end{proof}

Now we can complete the proof of Theorem~\ref{thm:generic}.
Proposition~\ref{prop:irreducible} implies that $\dim \varphi(X_0) > \dim \varphi(Y_0)$.
Since all fibers of $\varphi$ have the same dimension, $\dim X_0 > \dim Y_0$.
Hence, $\dim X \geq \dim (X_0 \setminus Y_0) = \dim X_0 > \dim Y_0 \geq \dim Y$.

\begin{remark}
  Theorem~\ref{thm:generic} is stated only for points satisfying~\hyperref[s1]{(S1)} and~\hyperref[s2]{(S2)}.
  However, the proof implies that every such point is generically nonsingular.
  We expect that the same technique can be used to prove that generically no removable singularities occur
  in points violating conditions~\hyperref[s1]{(S1)} and~\hyperref[s2]{(S2)}.
  This expectation agrees with our computational experiments with random operators and random polynomials.
  We think that these experimental results and Theorem~\ref{thm:generic} justify the choice $c = 1$ in Theorem~\ref{thm:curve}
  in most applications. 
\end{remark}

\begin{remark}
  On the other hand, neither Theorem~\ref{thm:generic} nor our experiments support the choice $c=1$ in the
  case $r_P=1$. Instead, it seems that in this case the cost for removability is systematically larger. To
  see why, consider the special case $P=y-x^2$ of substituting the polynomial $g(x)=x^2$ into a 
  solution $f$ of a generic operator~$L$. If the solution space of $L$ admits a basis of the form
  \begin{alignat*}5
    &1&&{} + a_{1,r_L}x^{r_L}&&{} + a_{1,r_L+1}x^{r_L+1} + \cdots,\\
    &x&&{} + a_{2,r_L}x^{r_L}&&{} + a_{2,r_L+1}x^{r_L+1} + \cdots,\\
    &\vdots\\
    &x^{r_L-1}&&{} + a_{r_L-1,r_L}x^{r_L}&&{} + a_{r_L-1,r_L+1}x^{r_L+1} + \cdots,
  \end{alignat*}
  and $M$ is the minimal operator for the composition, then its solution space obviously has the basis
  \begin{alignat*}5
    &1&&{} + a_{1,r_L}x^{2r_L}&&{} + a_{1,r_L+1}x^{2r_L+2} + \cdots,\\
    &x^2&&{} + a_{2,r_L}x^{2r_L}&&{} + a_{2,r_L+1}x^{2r_L+2} + \cdots,\\
    &\vdots\\
    &x^{2(r_L-1)}&&{} + a_{r_L-1,r_L}x^{2r_L}&&{} + a_{r_L-1,r_L+1}x^{2r_L+2} + \cdots,
  \end{alignat*}
  and so the indicial polynomial of $M$ is $\lambda(\lambda-2)\cdots(\lambda-2(r_L-1))$.
  According to the theory of apparent singularities~\cite{ince26,chenetal2016}, $M$~has a removable
  singularity at the origin and the cost of removability is as high as~$r_L$.

  More generally, if $g$ is a rational function and $\alpha$ is a root of~$g'$, so that
  $g(x)=c+\O((x-\alpha)^2)$, a reasoning along the same lines confirms that such
  an $\alpha$ will also be a removable singularity with cost~$r_L$.
\end{remark}


\noindent\textbf{Acknowledgement.} We thank the referees for their constructive critizism.

\bibliographystyle{plain}
\bibliography{bib}

\end{document}